\documentclass[conference]{IEEEtran}

\usepackage{cite}
\usepackage[cmex10]{amsmath}
\usepackage{array}
\usepackage{mdwmath}
\usepackage{mdwtab}
\usepackage[tight,footnotesize]{subfigure}
\usepackage{tikz}

\usepackage[T1]{fontenc}
\usepackage{amssymb,amsfonts,amsthm}
\usepackage{graphicx,epsfig,psfrag,color}
\usepackage{times}
\usepackage{dsfont}
\usepackage{algorithm}
\usepackage[noend]{algorithmic}
\usepackage{caption}

\usepackage{cleveref}
\usepackage{comment}

\usepackage{array}
\usepackage{mdwmath}
\usepackage{mdwtab}
\usepackage{tikz}
\usepackage{graphicx,epsfig,psfrag,color}
\usepackage{multirow}

\usetikzlibrary{decorations.pathreplacing}

\tikzstyle{vertex}=[circle,fill=black!15,minimum size=20pt,inner sep=0pt,font=\footnotesize]
\tikzstyle{smallvertex}=[vertex,minimum size=10pt]
\tikzstyle{operator}=[vertex,fill=black!1]
\tikzstyle{source} = [vertex, fill=red!34]
\tikzstyle{supersource} = [vertex, fill=blue!34,minimum size=15pt]
\tikzstyle{hiddensource} = [vertex, fill=red!8,minimum size=10pt]
\tikzstyle{smallsource} = [vertex, fill=red!34,minimum size=10pt]
\tikzstyle{receiver} = [vertex, fill=green!34]
\tikzstyle{smallreceiver} = [vertex, fill=green!34,minimum size=10pt]
\tikzstyle{edge} = [draw,thick,->]
\tikzstyle{undirect_edge} = [draw, thick]
\tikzstyle{dedge} = [edge,dotted]
\tikzstyle{redge} = [edge,color=red]
\tikzstyle{bedge} = [edge,color=blue]
\tikzstyle{gedge} = [edge,color=green]
\tikzstyle{medge} = [edge,color=magenta]
\tikzstyle{oedge} = [edge,color=orange]
\tikzstyle{bredge} = [edge,color=brown,line width=2pt]
\tikzstyle{weight} = [font=\footnotesize]
\tikzstyle{selected edge} = [draw,line width=5pt,-,red!50]

\newcommand{\bi}{\begin{itemize}}
\newcommand{\ei}{\end{itemize}}
\newcommand{\bal}{\begin{align}}
\newcommand{\eal}{\end{align}}

\newcommand{\de}{\stackrel{.}{=}}
\newcommand{\mc}{\mathcal}

\newtheorem{theorem}{Theorem}
\newtheorem{lemma}{Lemma}
\newtheorem{corollary}{Corollary}

\newtheorem{remark}{Remark}

\begin{document}

\title{Centralized vs Decentralized Multi-Agent Guesswork\vspace{-.15in}}

\author{\IEEEauthorblockN{Salman Salamatian}
\IEEEauthorblockA{MIT, USA
}
\and
\IEEEauthorblockN{Ahmad Beirami}
\IEEEauthorblockA{MIT, USA
}
\and
\IEEEauthorblockN{Asaf Cohen}
\IEEEauthorblockA{Ben-Gurion University, Israel
}
\and
\IEEEauthorblockN{Muriel M\'edard}
\IEEEauthorblockA{MIT, USA
}
\vspace{-.1in}
}

\maketitle

\begin{abstract}
We study a notion of guesswork, where multiple agents intend to launch a coordinated brute-force attack to find a single binary secret string, and each agent has access to side information generated through either a BEC or a BSC. The average number of trials required to find the secret string grows exponentially with the length of the string, and the rate of the growth is called the guesswork exponent. We compute the guesswork exponent for several multi-agent attacks. We show that a multi-agent attack reduces the guesswork exponent compared to a single agent, even when the agents do not exchange information to coordinate their attack, and try to individually guess the secret string using a predetermined scheme in a decentralized fashion. Further, we show that the guesswork exponent of two agents who do coordinate their attack is strictly smaller than that of any finite number of agents individually performing decentralized guesswork. 
\end{abstract}



\vspace{-.2cm}
\section{Introduction}
\label{sec:intro}

We consider a setup where a system is protected using a password $X^n \in \mc{X}^n$, drawn i.i.d. at random from a distribution $p_X(\cdot)$ on the finite alphabet $\mc{X}$. An adversary wishes to breach the system by guessing the password. Assuming $n$ is known to the adversary, a brute-force attack on the system would consist of first producing a list of all of the $|\mc{X}|^n$ strings in $\mc{X}^n$ ordered from the most likely to the least likely with respect to $p_{X^n}(\cdot)$, and then exhausting the list one by one until successfully guessing the password. Let the guesswork, denoted by $G(X^n)$, be defined as the position at which the password string $X^n$ appears in the adversary's list of all strings. The guesswork $G(X^n)$ can  be thought of as the computational cost in terms of number of queries required of an adversary to breach the system. 
As shall be discussed, $G(X^n)$ grows exponentially with $n$ for the processes considered in this paper, and the rate of its growth is referred to as the guesswork exponent.

If $m$ adversarial agents coordinate their attack on the secret string, the system will be compromised as soon as either of them succeeds, and hence, the average guesswork is reduced. Indeed, an optimal strategy would consist here of having each agent query the most likely sequence that has not yet been queried by any of the other  agents.  As the length of the password $n$ grows, the impact of finitely many agents becomes more and more negligible, and since the size of the list grows exponentially in $n$, dividing the list by a constant does not change the guesswork exponent.

In this work, we further assume that the agents have access to a side information string $Y^n$, which they use to construct an updated list of strings, this time ordered with respect to $p_{X^n|Y^n}(\cdot|Y^n)$.  In its most general form, this side information can model complex additional information that the adversary may have acquired on the choice of the password, ranging from background search on the user who chose the password, to simply \emph{behind the back} attacks in which an illegitimate person observes parts of the password. For example, considering $Y^n$ to be the output of a binary erasure channel can model an agent who has acquired parts of the secret password in the clear. Consider now a case in which multiple adversaries try to guess the password, each having access to some side information $Y^n_{(i)}$, which is assumed to be generated independently given $X^n$ through some discrete memoryless channel. Contrary to the case where there is no side information, we demonstrate that having even a fixed number of agents can help in reducing the exponent of the guesswork --- whether they coordinate and use their side information in a centralized manner, or try independently in a decentralized way to guess the password (see Fig.~\ref{fig:1}). We illustrate the impact of multiple agents by studying both the centralized and the decentralized mechanisms for side information provided through the binary symmetric channel ($\texttt{BSC}$) and the binary erasure channel ($\texttt{BEC}$). 

This setting can also indirectly model adversaries and users over multiple accounts, some of which have been compromised. Suppose a user has several accounts, each requiring a password. The user may decide to use one identical password for all of the accounts, where the compromise of one of the accounts puts in peril all of his accounts. On the other extreme, he may decide to use completely independent passwords for each of the accounts, in which case one password being compromised does not give away any information on any of the other passwords. In practice, most users   settle for a solution in between these two extremes. For example, the user may choose to slightly tweak their passwords from one account to another as to avoid the disastrous consequences of one account being compromised providing access to the rest of the accounts, while still maintaining some convenience. In this case, if one password is compromised, an adversary gains some side-information about the rest of the passwords.

The normalized moments of guesswork are of great interest as they provide operational meanings in several information theoretic problems.
For any $\alpha >0$, let $E_\alpha(p_X)$ denote the guesswork exponent and be defined as 
$$
E_\alpha(p_X) : = \frac{1}{\alpha}\lim_{n \to \infty} \frac{1}{n} \log E\{ [G(X^n)] ^ \alpha\},
$$
where the expectation is with respect to the measure $p_X$.
Further, let $E_0(p_X): = \lim_{\alpha \to 0} E_\alpha(p_X).$
For example, $E_1(p_X)$ is the exponential  growth rate of the expected number of queries required of the adversary to breach the secret string, and $E_0(p_X)$ is the average codeword rate in optimal one shot source coding \cite{Kosut, ITW14}.
Similarly, one can extend these notions to guesswork with side-information.
The conditional guesswork, denoted $G(X^n|Y^n)$, can be thought of as the computational cost of an agent who has acquired side information $Y^n$.  The conditional guesswork exponent $E_\alpha(p_{X,Y})$ then describes the exponential rate of conditional guesswork.

\noindent \textbf{Related Work:} We briefly mention some related work. Guesswork was first considered in \cite{Messay}, where it was shown that guesswork is not necessarily related to the Shannon entropy. In \cite{Arikan}, it is shown that the moments of guesswork for i.i.d. sequences are related to the R\'enyi-entropy of the source. Since then, this was generalized to various source processes (see~\cite{malone2004guesswork,pfister2004renyi}), and under source uncertainty in \cite{Sundaresan}. In \cite{ChristiansenDuffy}, guesswork is shown to satisfy a large deviation principle. \cite{ArikanMerhav} studies guesswork subject to distortion. A geometric perspective on guesswork is introduced in \cite{Beirami2}. Guesswork, as a metric for quantifying the computational effort of brute-force attacks has been studied under various settings: under an entropy constraint in \cite{Beirami}, over the typical set in \cite{christiansen2013brute}, multiple users in \cite{christiansen2015multi}, with erasures in \cite{christiansen2013guessing}.
Guesswork is central to several other problems in information theory, ranging from the computational cost of sequential decoding \cite{Arikan}, to the error exponent in list decoding~\cite{Merhav-list}.

\noindent \textbf{Main Contribution:} In this paper, we consider the guesswork exponent under two types of side information, namely $\texttt{BEC}_{\epsilon}$ and $\texttt{BSC}_\delta$, where $\epsilon$ and $\delta$ are the respective channel parameters. We characterize the impact of multiple agents in this setting, and show that even a finite number of agents reduces the conditional guesswork exponent. We carry this out by considering two extreme settings, one in which the agents are guessing the password, individually and independently (decentralized mechanism), and one in which all the side information is collected and used collectively (centralized mechanism). Section~\ref{sec:prelim} introduces the setting along with some notations and background on guesswork with side information. Results for the $\texttt{BEC}_\epsilon$ and $\texttt{BSC}_\delta$ are presented in Section~\ref{sec:BEC} and Section~\ref{sec:bsc}, respectively.

\begin{figure}
	\centering
	\includegraphics[width = .48\textwidth]{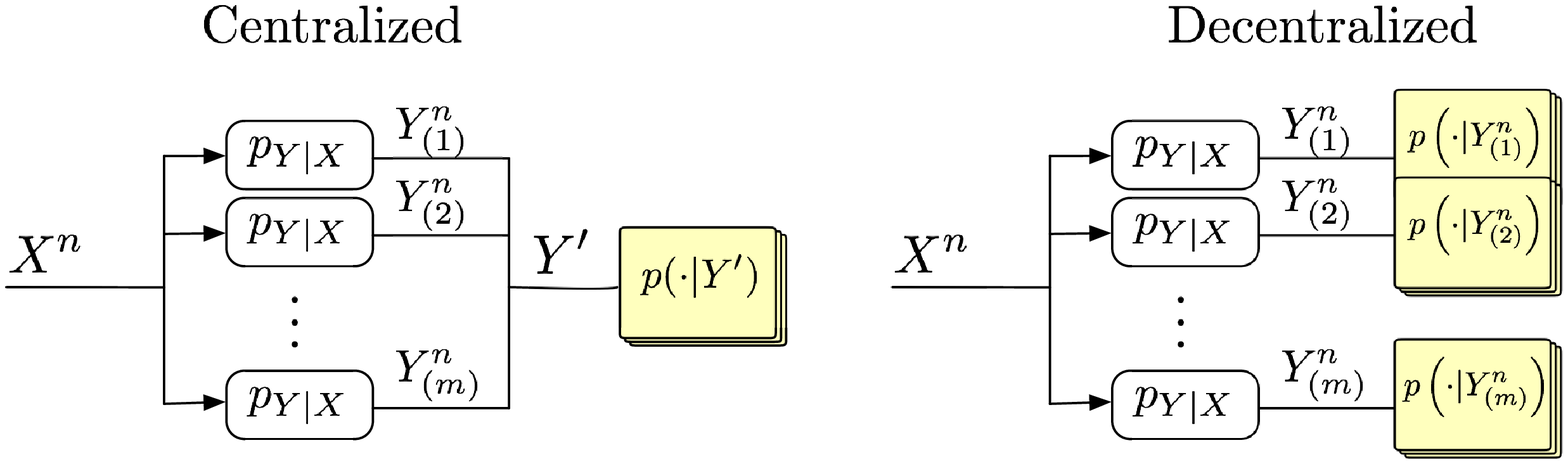}
	\caption{In the centralized mechanism, a single list is constructed by collecting all the side-informations. In the decentralized setting, each agent constructs a separate list.}\label{fig:0}
	\vspace{-.2in}
\end{figure}

\section{Preliminaries}\label{sec:prelim}

\subsection{Notations}
Let $(X^n,Y^n) := (X_1,Y_1),\ldots,(X_n,Y_n)$, where $(X_i,Y_i) \in \mathcal{X}\times \mathcal{Y}$, denote a random string of length $n$ drawn i.i.d. from a distribution $p_{X,Y}$ over some finite alphabet $\mathcal{X} \times \mathcal{Y}$. 
The sequence $X^n$ can be thought of as the password to guess, while the sequence $Y^n$ can be thought of as side information. The conditional guesswork $\mathbb{E}\left[ G(X^n|Y^n) \right]$ is then the computational cost of the adversary with side information $Y^n$.
For $\beta>0, \beta \neq 1$, we denote by $H_\beta(X)$ and $H_\beta(X|Y)$, respectively, the R\'enyi-entropy and conditional R\'enyi-entropy of order $\beta$, defined in the usual way:
\begin{eqnarray*}
	H_{\beta}(X) &=& \frac{\beta}{1-\beta} \log \left( \sum_{x} p_X(x)^\beta \right)^{1/\beta}, \nonumber \\
	H_{\beta}(X|Y) &=& \frac{\beta}{1 - \beta} \log \left( \sum_{y} \left( \sum_{x} p_{X,Y}(x,y)^\beta \right)^{1/\beta}\right).
\end{eqnarray*} 
We will focus on the case of binary input alphabets, \emph{i.e.}, $\mathcal{X} = \{ 0,1\}$.
For $0 \leq p \leq 1/2$, we denote by $H_{\beta}(p)$ the binary R\'enyi entropy of order $\beta$, and by $H(p)$ the binary Shannon entropy. Furthermore, we let $D(p||q)$ be defined as the KL-divergence between two binary distributions parameterized by $p$ and $q$, respectively, that is:
\begin{align}
D(p||q) = p\log \frac{p}{q} + (1-p)\log \frac{1-p}{1-q} .
\end{align}
Given an observation $Y^n = y^n$, we denote by $G(X^n|Y^n = y^n)$ the position of $X^n$ in the list of ordered sequences $x^n$ from most likely to least likely according to $p_{X^n|Y^n}(\cdot | y^n)$. The conditional Guesswork $\mathbb{E} \left[ G(X^n | Y^n)^\alpha \right]$ is then the average $\sum_{y^n}  p_{Y^n}(y^n)\mathbb{E} \left[ G(X^n | Y^n = y^n )^\alpha \right]$. We are interested in the conditional guesswork exponent defined as 
\begin{align}
E_\alpha(p_{X,Y}) :=  \lim_{n \to \infty} \frac{1}{n} \log \mathbb{E}[G(X^n|Y^n)^\alpha],
\end{align}
for $\alpha > 0$. An application of L'Hopital's rule yields the following useful equality:
\begin{align}
\lim_{\alpha \to 0} \frac{1}{\alpha} E_\alpha(p_{X,Y}) = \lim_{n \to \infty} \frac{1}{n} \mathbb{E} \left[ \log(G(X^n|Y^n))\right].
\end{align}
In a seminal result, Ar{\i}kan \cite{Arikan} showed that the moments $\alpha$ of guesswork are related to the Renyi entropies of order $\frac{1}{1+ \alpha}$ of the source, that is:
\begin{align}
E_\alpha(p_{X,Y}) = \alpha H_{\frac{1}{1+\alpha}}(X|Y).
\end{align}
When the input distribution $p_X$ is clear from context, we may write $E_\alpha(p_{Y|X})$. We use $f(n) \de g(n)$, if $\lim_{n \to \infty} \frac{\log f(n)}{\log g(n)} = 1$. Logarithms and exponents are in base 2.

\subsection{Background on Noise and Erasures}

For the remainder of the paper, we will suppose that $X^n$ is a uniform Bernoulli sequence, and we will be interested in two families of side information. Namely, we will let $Y^n$ be the output of $X^n$ through a binary symmetric channel ($\texttt{BSC}_\delta$), or through a binary erasure channel ($\texttt{BEC}_\epsilon$). We will use the notation $E_\alpha(\texttt{BSC}_\delta)$ and $E_\alpha( \texttt{BEC}_\epsilon)$, to denote each corresponding exponent, where it is implicit that the input distribution $p_{X^n}$ is chosen to be uniform over binary sequences of length $n$.

\noindent \textbf{BSC:} Let $Y^n$ be the output of $X^n$ through a \texttt{BSC} with flip-over probability $\delta \leq 1/2$. Noting that $X^n = Y^n + Z^n$, where the addition operation is over $\mathbb{Z}_2$, it is easy to see that $G(X^n|Y^n) = G(Z^n)$, and the average guesswork is given by:
\begin{align}
E_\alpha(\texttt{BSC}_\delta) =  \lim_{n \to \infty} \frac{1}{n} \log \mathbb{E}[G(Z^n)^\alpha] = \alpha H_{1/(1+\alpha)}(\delta).
\end{align}

\noindent \textbf{BEC:} Let $Y^n$ be the output of $X^n$ through a \texttt{BEC} channel with erasure probability $0 \leq \epsilon \leq 1$. Denote by $\mathcal{E}_n$ the number of erasures. Then, we have that $G(X^n|Y^n) = G(X'^{\mathcal{E}_n})$, where $X'^{\mathcal{E}_n}$ is the erased sequence. It has been shown in \cite{christiansen2013guessing}, using results from large deviation theory, that the $\alpha$-th moment of guesswork in this setting (referred to as \emph{subordinated Guesswork} in \cite{christiansen2013guessing}) is: 
\begin{align}
E_\alpha(\texttt{BEC}_\epsilon) = \sup_{\lambda \in [0,1]} \left( \alpha \lambda -  D(\lambda||\epsilon) \right) .  
\end{align}
Specifically, for $\alpha = 1$, the exponent of the average guesswork is given by:
\begin{align}
\lim_{n \to \infty} \frac{1}{n} \log \mathbb{E}[G(X^{\mathcal{E}_n})] = \log \left( 1 + \epsilon \right).
\end{align}

Finally, the following lemma which we will use in the proofs, characterizes the guesswork exponent of a sequence generated by the concatenation of a uniform binary sequence, and an arbitrary \emph{i.i.d.} sequence.

\begin{lemma}\label{lem:concatenation_uniform}
	Let $U \sim Ber(1/2)$ and $V \sim Ber(p)$, with $p \leq 1/2$, and denote by $U^{m_n}$ and $V^{n-m_n}$ their \emph{i.i.d.} sequences, for some sequence $m_n$ such that $\lim_{n \to \infty} \frac{m_n}{n} = \lambda$. Then, the guesswork exponent for the sequence $X^n = (U^{m_n}, V^{n - m_n})$ obtained by the concatenation of $U^{m_n}$ and $V^{n-m_n}$ is:
	\begin{align}
	\lim_{n \to \infty}\frac{1}{n} \log \mathbb{E} \left[ G(X^n)^\alpha \right] = \lambda \alpha + (1 - \lambda) \alpha H_{1/1+ \alpha} (p).
	\end{align}
\end{lemma}
\begin{proof}[Proof Sketch]
	The result follows from the fact that we need to guess the subsequence $V^{n-m_n}$, but each such subsequence has $2^{m_n}$ uniform possibilities for $U^{m_n}$.
\end{proof}

\subsection{Setting}

As shown above, the problem of Guesswork under side information is well understood. A more complicated problem is one in which multiple agents receive side information, and not a single source of side information. Precisely, let there be $m$ agents, each observing an independent realization of a side information $Y_{(i)}^n, i = 1,\ldots,m$, where $Y^n_{(i)}$ is the output of the password sequence $X^n$ through a discrete memory-less channel.
Clearly, if all the agents cooperate and share their side information, they can construct an optimal list based on the aggregate collection of side information $Y' = (Y_{(1)}^n, \ldots, Y_{(m)}^n)$. This strategy clearly out performs the strategy in which each agent tries to guess the sequence on its own. However, it is not clear to which extent this sharing of side information improves the exponent with respect to a decentralized approach. To answer this question, we consider the two families of side information we already introduced, namely $\texttt{BEC}_\epsilon$ and $\texttt{BSC}_\delta$, and characterize the conditional guesswork exponent under the two following strategies, illustrated in Fig.~\ref{fig:0}:

\noindent \textbf{Decentralized Mechanism:} Each of the $m$ agents tries to guess $X^n$ based on its own observation $Y_{(i)}^n$. The process ends when at least one of the agents correctly guesses $X^n$. The conditional guesswork exponent for this strategy, denoted $E^{(\text{d})}_\alpha(p_{Y|X}^m)$, is therefore:
\begin{equation}
\frac{1}{\alpha} E^{(\text{d})}_\alpha(p_{Y|X}^m) = \frac{1}{\alpha} \lim_{n \to \infty} \frac{1}{n}\log \mathbb{E}\left[ \min_{i = 1, \ldots,m} \left\{ G(X^n | Y^n_{(i)})^\alpha\right\} \right].
\end{equation}

\noindent \textbf{Centralized:} The agents share their observations $Y_{(i)}^n$, $i = 1,\ldots,m$ with a central authority who collapses the side information and constructs an optimal list based on $p_{X|Y_{(1)},\ldots,Y_{(m)}}$. The conditional guesswork exponent for this strategy is denoted by $E^{\text{(c)}}_\alpha(p_{Y|X}^m)$, and:
\begin{equation}
E^{(\text{c})}_\alpha(p_{Y|X}^m) = E_\alpha(p_{Y'|X}),
\end{equation}
where $Y' = (Y_{1}, \ldots, Y_{(m)})$ and $p_{Y'|X}(y_1,\ldots,y_m | x) = \prod_{i = 1}^m p_{Y|X}(y_i|x)$.

Note that it follows directly that $E^{\text{(d)}}_\alpha(p_{Y|X}^1) = E^{(\text{c})}_\alpha(p_{Y|X}^1) = E_\alpha(p_{Y|X})$.

In the rest of the paper, we will characterize the conditional guesswork exponents under $\texttt{BEC}_\epsilon$ and $\texttt{BSC}_\delta$ side information. Precisely, we let $Y_{(1)}^n, \ldots, Y_{(m)}^n$ be the output of $X^n$ through $m$ independent $\texttt{BSC}_\delta$ or $\texttt{BEC}_\epsilon$ channels. Note that, even though the initial channel is a simple binary channel, the resulting channel from the collapsing of the side information may be more complex. This will be the case for \texttt{BSC}. In the next section, we analyze the guesswork exponent for the \texttt{BEC} side information. The analysis for the \texttt{BSC} is in Section~\ref{sec:bsc}.

It has to be noted that we are studying asymptotic behaviors for fixed $m$, that is $m$ does not grow with $n$. In the sequel, we may take the limit when $m \to \infty$ and determine say $\lim_{m \to \infty} E_\alpha^{\text{(c)}}(p_{Y|X}^m)$, where $E_\alpha^{\text{(c)}}(p_{Y|X}^m)$ is itself the result of a limit when $n \to \infty$. It is understood here that the order of the limits is crucial and an interchange of limit is not possible.

\section{BEC}\label{sec:BEC}

\subsection{Centralized Mechanism}

The $\texttt{BEC}_\epsilon$ is simple to analyze because collapsing information is tractable. In particular, the symbol in position $i$ in the sequence $X^n$ is erased in all received signals $Y_i^n$ with probability $\epsilon^m$. Therefore, the resulting collapsed random variable $\tilde{Y}^n$ is the output of $X^n$ through a $\texttt{BEC}$ with erasure probability $\epsilon^m$, and we have the following.

\begin{theorem}
	The guesswork exponent for the centralized Mechanism with $m$ agents under \texttt{BEC} is:
	\begin{align}
	E^{(c)}_\alpha(\texttt{BEC}_\epsilon^m) = \max_{\lambda \in [0,1]} \left( \alpha \lambda -  D(\lambda \| \epsilon^m) \right).
	\end{align}
\end{theorem}

Carrying out the maximization for $\alpha = 1$, we have the following.
\begin{corollary}
	The centralized Mechanism with $m$ agents under \texttt{BEC} side information has expected Guesswork exponent (see Fig.~\ref{fig:1}):
	\begin{align}
	E^{(c)}_1(\texttt{BEC}_\epsilon^m) = \log \left( 1 + \epsilon^m \right).
	\end{align}
\end{corollary}

\begin{remark}
	The function $f(x) = \log(1 + x^m)$ over $x \in [0,1]$, is convex for any $m \geq 2$. Moreover, as the number of agents increases, the exponents tends towards a flat function $E^{\text{(c)}}_\alpha = 0$, with a discontinuity at $\epsilon = 1$. Moreover, since the first derivative (when $\alpha = 1$) is $m\frac{\epsilon^{m-1}}{1+\epsilon^{m}}$ for any $m \geq 2$, the centralized curve starts flat with a negligible exponent for small $\epsilon$.
\end{remark}

\begin{figure}[t]
	\centering
	\includegraphics[scale=.620]{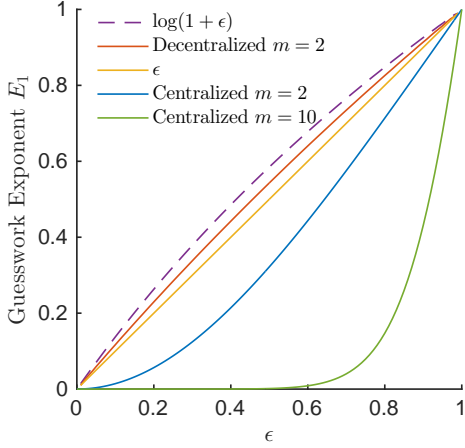}
	\caption{Comparison of centralized and decentralized settings for the \texttt{BEC}.}\label{fig:1}
	\vspace{-.15in}
\end{figure}

\subsection{Decentralized Mechanism}

The study of the decentralized case is more involved since, on the one hand, one cannot construct a unified list based on all $\left\{ Y^n_{(i)}\right\}_{i=1}^m$, yet, on the other hand, the guesswork random variables $G\left(X^n | Y^n_{(i)}\right)$  are not independent and one cannot easily combine $\left\{ G\left(X^n|Y^n_{(i)}\right)\right\}_{i = 1}^m$. First, we discuss the result:

\begin{theorem}\label{thm:distrib_BEC}
	The decentralized mechanism with \texttt{BEC} side-information has Guesswork Exponent:
	\begin{align}
	E_\alpha^{(d)}(\texttt{BEC}_\epsilon^m)= \sup_{\lambda \in [0,1]} \left( \alpha \lambda - m D(\lambda || \epsilon) \right).
	\end{align}
	\vspace{-.2in}
\end{theorem}
\noindent Before we proceed to the proof, some remarks are in order. One can verify that the limit of the Guesswork Exponent for the decentralized mechanism, as the number of agents $m$ increases, converges towards $\epsilon$ (see Fig~\ref{fig:1}). Indeed, for large $m$, the term $-mD(\lambda\|\epsilon)$ dominates, and the solution of the optimization is $\lambda \simeq \epsilon$. On the other hand, Remark~1 establishes that the Guesswork exponent is convex for any $m \geq 2$, implying that even two agents that collapse their side information are more powerful than any finite number of agents guessing $X^n$ in a decentralized way. Note that this claim has to be nuanced. Indeed, we are looking at the asymptotic behavior of the guesswork exponent as $n \to \infty$, for a fixed number of agents, i.e., this does not allow a growing number of agents with $n$.

\begin{proof}[Proof Sketch]
	The proof of Thm~\ref{thm:distrib_BEC} follows from two steps. First, we establish an upper bound based on the shortest sequence. Due to space restrictions, we provide below only a proof sketch in the case of $m=2$.
	First, we find an upper bound on the guesswork exponent by considering the exponent of the shortest sequence. The details are omitted, but follow from a standard use of the method of types.
	\begin{align}
	\mathbb{E}[\min_{i = 1,\ldots,m} \{ G(X^{\mathcal{E}^{(i)}_n})^\alpha\}] \leq \mathbb{E}[G(X^{\mathcal{E}^*_n})^\alpha].
	\end{align}
	where $\mathcal{E}^*_n$ is the random variable representing the minimum number of erasures among all $m$ agents. Therefore, we have:
	\begin{align}
	E_\alpha^{(d)} (\texttt{\texttt{BEC}}^m_\epsilon) \leq \sup_{\lambda \in [0,1]}\left( \alpha \lambda - m D(\lambda \| \epsilon) \right).
	\end{align}
	To obtain a matching lower-bound, we consider an oracle that provides additional information to both agents, strictly reducing their guesswork. In general terms, the additional information from the oracle allows to construct explicitly the optimal list of both agents. More precisely, this is achieved by transmitting the position of the common erasures for both agent. The optimal joint strategy is then to construct lists as to minimize queries that have a common subsequence in the overlapping erasures. Indeed, each incorrect query from an agent, shapes the probability distribution of the second agent because of the common sequences. We show that this probability shaping, can be again lower-bounded by a mechanism in which each agent has two guesses at each step, instead of one, therefore not affecting the guesswork exponent.
\end{proof}

\section{BSC}\label{sec:bsc}

\subsection{Centralized Mechanism}

In the case of the $\texttt{BSC}_\delta$, the centralized mechanism is more involved to analyze. Indeed, the resulting channel $\texttt{BSC}_\delta^m$ is not a $\texttt{BSC}$ anymore, since one has $m$ noisy measurements per password-bit. Indeed, as it will be clear soon, guessing should be preceded with some kind of estimation. Nevertheless, for $m = 2$, we can characterize precisely what this channel exactly is, by considering the $2^m = 4$ cases. We will then discuss how to generalize this result to arbitrary $m > 2$.

\begin{theorem}
	The centralized mechanism with $m = 2$ agents under $\texttt{BSC}_\delta$ side-information satisfies:
	\begin{align}
	E_\alpha^{(c)}(\texttt{BSC}_\delta^2) = \sup_{\lambda \in [0,1]} \left( \alpha\lambda H_{1/1+\alpha} \left( \frac{\delta^2}{1 - 2\delta(1-\delta)}\right) + \right. & \nonumber \\
	\left. \phantom{\lambda H_{1/1+\alpha} \left( \frac{\delta^2}{1 - 2\delta(1-\delta)}\right)} \alpha (1-\lambda) - D\left(\lambda \| 2\delta(1-\delta) \right) \right)& \nonumber .
	\end{align}
\end{theorem}

\begin{corollary}
	The average guesswork, when $\alpha = 1$, is (Fig.~\ref{fig:2}) 
	\begin{align}
	E_1^{(c)}(\texttt{BSC}_\delta^2) = \log(4\delta(1 - \delta) + 1).
	\end{align}
\end{corollary}

\begin{proof}
	Denote by $Y^n_1$ and $Y^n_2$ the sequence of side information observed by each agent, and divide each into two parts. In the first part, $Y_1^n$ and $Y_2^n$ agree and have the same bit in every position, that is on this subsequence, the centralized $\tilde{Y}^n$ is essentially the result of a \texttt{BSC} with parameter $\delta^2/(1 - 2\delta(1-\delta))$. In the second part, they disagree and have contradicting bits in every position, which is essentially an erasure. We let $\lambda \in [0,1]$ be the fraction of bits over which they agree, \textit{i.e.} $\lambda n$ is the size of the first subsequence defined above. Therefore, the central authority has to guess a sequence of the type $\tilde{X}^n = (\tilde{U}^{n(1- \lambda)}, \tilde{Z}^{n \lambda})$, where $\tilde{U}^{n(1- \lambda)}$ is an \emph{i.i.d.} sequence of uniform Bernoulli random variables that correspond to the erasures, and $Z^{n \alpha}$ is an \emph{i.i.d.} sequence of Bernoulli random variables with parameter $\tilde{\delta} \triangleq \delta^2/(1 - 2\delta(1-\delta))$. By Lemma~\ref{lem:concatenation_uniform}, we have that:
	\begin{align}
	\lim_{n \to \infty} \frac{1}{n} \log \mathbb{E}[G(\tilde{X}^n)^\alpha] = \lambda \alpha + (1 - \lambda)\alpha H_{1/1+\alpha}(\tilde{\delta}).
	\end{align}
	Noting that the probability of the subsequence of agreements to be of length $\lambda n$ is (up to polynomial factors) $\exp \left\{-n D(\lambda \| 2\delta(1 - \delta)) \right\}$, we get the desired optimization. Solving for $\alpha = 1$ yields the corollary.
\end{proof}

Note that one can easily verify the following
\begin{align}
\log(4\delta(1 -\delta) + 1) \leq H(\delta),
\end{align}
with equality only if $\delta = 1/2$ or $\delta = 0$.

The previous theorem only treats the case of $m = 2$ agents, although a similar technique can be used to tackle any $m \geq 2$ number of agents. Unfortunately, this method is intractable for large $m$. However, the following result allows us to compute the limit as the number of agents grows to infinity:
\begin{lemma}
	Let $\delta < \frac{1}{2}$, then:
	\begin{align}
	\lim_{m \to \infty} E_{\alpha}^{(c)} (\texttt{BSC}_\delta^m) = 0.
	\end{align}
\end{lemma}
\begin{proof}
	For a fixed $n$ and $m$, we do a deterministic pre-processing on the sequences $Y^n_{(1)}, \ldots, Y^n_{(m)}$, which can only increase the guesswork, by definition. Namely, we let $\tilde{Y}_i$ be defined as the majority bit among the received side-informations at index $i$, that is :
	\begin{align}
	\tilde{Y}_i = \left\{
	\begin{array}{ll}
	0 & \text{,if } N_i(0) \geq N_i(1), \\
	1 & \text{,if } N_i(0) < N_i(1),
	\end{array} \right.
	\end{align} 
	where $N_i(0) = \sum_{j = 1}^m Y_{(j),i}$, for $Y_{(j),i}$ the $i$-th bit of the sequence $Y^n_{(j)}$, and $N_i(1) = n - N_i(0)$. Then, it is easy to see that $\tilde{Y}^n$ is the output of $X^n$ through a \texttt{BSC} with parameter $\delta_m$, such that $\delta_m \to 0$ as $m \to \infty$ for any $\delta < 1/2$. Therefore, we have, for any $n$, and for fixed $m$, the following inequality:
	\begin{align}
	\mathbb{E}[G(X^n|Y')^\alpha] &\leq \mathbb{E}[G(X^n|\tilde{Y}^n)^\alpha] \\
	\Rightarrow E_\alpha^{\text{(c)}} (\texttt{BSC}_\delta^m) &\leq E_\alpha(\texttt{BSC}_{\delta_m}) \\
	\Rightarrow \lim_{m \to \infty} E_\alpha^{\text{(c)}} (\texttt{BSC}_\delta^m) &\leq \lim_{m \to \infty} E_\alpha(\texttt{BSC}_{\delta_m}).
	\end{align}
	As the right hand side of the last inequality converges to $0$ for any $\delta < \frac{1}{2}$, we obtain the desired result.
\end{proof}
In other words, when $m$ is large enough, one can \emph{estimate} each bit of the password based on the noisy observations.

\subsection{Decentralized Mechanism}

\begin{figure}
	\centering
	\includegraphics[scale = .62]{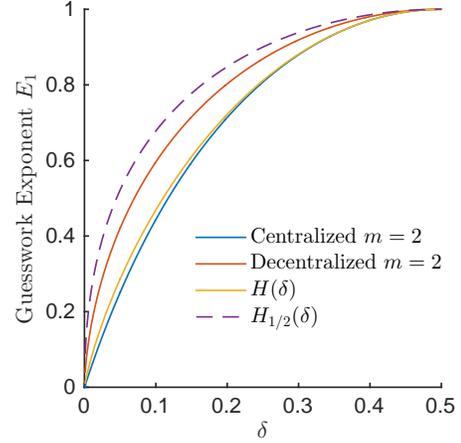}
	\caption{Comparison of the centralized and decentralized setting for $\texttt{BSC}$.}\label{fig:2}
	\vspace{-.1in}
\end{figure}

In contrast with the $\texttt{BEC}$ case, when the side-information $Y^n_{(i)}$ is the result of a $\texttt{BSC}$, the resulting guessworks are independent. Indeed, as stated before  $G(X^n|Y_{(i)}^n) = G(Z^n_{(i)})$, where now the sequences of flips $Z^n_{(i)}$ are independent. The following result, which is a special case of the more general large deviation result in \cite{christiansen2015multi} follows directly:

\begin{theorem}
	The decentralized mechanism with $m$ agents under \texttt{BSC} side-information has expected Guesswork exponent:
	\begin{align}
	E_\alpha^{(d)}(\texttt{BSC}_\delta^m) = \alpha H_{\frac{m}{\alpha+m}}(\delta).
	\end{align}
\end{theorem}
\begin{proof}[Alternative Proof]
	For completeness, we provide a proof that does not require to evaluate the full large deviation behavior of the guesswork to evaluate its moments. First we recall the following elementary result. Let $S^n_i$ be the sum of $n$ i.i.d. coin flips with parameter $\delta$. Then, for any $\delta < s \leq 1$:
	\begin{align}
	&\text{Pr}(\min_{i = 1\ldots,m} S_i = sn) = m \text{Pr}(S_1 = sn) \prod_{i=2}^mPr(S_i \geq sn) \\
	\quad &\de \exp \{- n D(s || \delta)  \} \left( \exp\{ -nD(s||\delta)\} \right)^{m-1} \\
	\quad &\de \exp \{ -nm D(s||\delta)\}.
	\end{align}
	Alternatively, when $0 < s \leq \delta$, we have:
	\begin{align}
	\text{Pr} \left(\min_{i = 1,\ldots,m} S_i = sn\right) & \de \exp\{ -nD(s||\delta)\}.
	\end{align}
	Using the previous results, and recalling that $G(Z^n_{(i)}) \de 2^{S^n_i}$, where $S^n_i$ is the number of 0's in the sequence (the type of the binary sequence), we obtain that:
	\begin{align}
	\mathbb{E}[\min_{i = 1,\ldots,m} G(Z^n_{(i)})^\alpha] \de \exp \left\{ n \sup_{\lambda \in [0,1]} \left(  \alpha \lambda - f(\lambda,m) \right) \right\},
	\end{align}
	where $f(\lambda,m) = \mathbf{1}\{ \lambda > \delta \} mD(\lambda ||\delta) + \mathbf{1}\{\lambda \leq \delta \} D(\lambda || \delta)$. The desired result follows by observing that the maximization over $\lambda$ always lead to a solution in the range $\lambda > \delta$, for any $\alpha > 0$.
\end{proof}

\begin{remark}
	The limit when $m \to \infty$ of the decentralized setting tends to the Shannon entropy $\alpha H(\delta)$ for any $\alpha>0$.
\end{remark}

\section*{Acknowledgment} 
\noindent The authors are thankful to Ken Duffy, whose comments greatly improved the presentation and content of this paper.

\bibliographystyle{IEEEtran}
\bibliography{IEEEabrv,strings}

\end{document}